\documentclass[11pt]{llncs}

\usepackage{color}
\usepackage{url}

\usepackage[usenames,dvipsnames]{xcolor}
\usepackage{multirow}
\usepackage{epsfig}
\usepackage{graphicx}
\usepackage{a4wide}
\usepackage{pdfpages}
\usepackage{enumerate}

\newcommand{\TBR}{{\rm TBR}}
\newcommand{\MP}{{\rm MP}}

\newtheorem{observation}[theorem]{Observation}

\newcommand{\blue}[1]{\textcolor{black}{#1}}

\newcommand{\steven}[1]{\textcolor{black}{#1}}

\definecolor{applegreen}{rgb}{0.55, 0.71, 0.0}

\newcommand{\polish}[1]{\textcolor{black}{#1}}

\newcommand{\purple}{\textcolor{black}}

\usepackage{amsmath}
\usepackage{graphicx}

\title{A tight kernel for computing the tree bisection and reconnection distance between two phylogenetic trees }
\author{Steven Kelk\inst{1}, Simone Linz\inst{2}}

\institute{Department of Data Science and Knowledge Engineering (DKE),\\ Maastricht University, The Netherlands,\\ \email{steven.kelk@maastrichtuniversity.nl}
\and Department of Computer Science, University of Auckland, New Zealand,\\
\email{s.linz@auckland.ac.nz}}

\providecommand{\keywords}[1]{\textit{Keywords:} #1}

\begin{document}
\maketitle

\begin{abstract}
In 2001 Allen and Steel showed that, if subtree and chain reduction rules have been applied to two
unrooted phylogenetic trees, the \steven{reduced} trees will have at most $28k$ taxa where $k$ is the TBR (Tree Bisection and Reconnection)
distance between the two trees. Here we reanalyse Allen and Steel's kernelization algorithm and prove that the
reduced instances will in fact have at most $15k-9$ taxa. Moreover we show, by describing a family of instances which have exactly $15k-9$ taxa after reduction, that this new bound is tight. These instances also have no common clusters, showing that a third commonly-encountered reduction rule, the cluster reduction, cannot further reduce the size of the kernel in the worst case. To achieve these results we introduce and use ``unrooted generators'' which are analogues of rooted structures
that have appeared earlier in the phylogenetic \purple{networks} literature. Using similar
argumentation we show that, for the \polish{minimum hybridization problem} on two rooted trees, $9k-2$ is a tight bound (when subtree and chain reduction rules have been applied) and $9k-4$ is a tight bound (when, additionally, the cluster
reduction has been applied) on the number of taxa, where $k$ is the \polish{hybridization number} of the two trees. 
\end{abstract}


\section{Introduction}

In the study of evolution \emph{phylogenetic trees} are often used to depict the evolution of a set of species \polish{(or more abstractly \emph{taxa})} $X$. These are trees in the usual graph-theoretical sense where the leaves are bijectively labelled by $X$ and the internal \purple{vertices} represent hypothetical (common) ancestors of $X$ \cite{steel2016phylogeny}. Phylogenetic trees are typically constructed from genetic markers, such as DNA alignments, and under many objective functions the goal of constructing the ``best'' phylogenetic tree is NP-hard \cite{money2012characterizing}. This has stimulated the development of heuristics which explore the space of phylogenetic trees using topological rearrangement moves \cite{felsenstein2004inferring,john2017shape}. One popular such move is the \emph{Tree Bisection and Reconnection} (TBR) move which, informally, deletes some edge of the tree and then re-attaches the two induced components by a newly introduced edge. In attempting to understand the connectivity of phylogenetic tree space under the action of TBR moves, it is natural to ask the following: what is the smallest number of TBR moves required to transform one tree $T$ into another tree $T'$? This is known as the  \blue{TBR {\it distance}} of the two trees, denoted $d_\TBR(T,T')$. It is NP-hard to compute \cite{AllenSteel2001,hein1996complexity}. In 2001 Allen and Steel showed the following kernelization result: if \emph{common pendant subtrees} and \emph{common chains} in the two trees are repeatedly collapsed, then this preserves $d_\TBR$, and upon termination the reduced trees will contain at most $28 \cdot d_\TBR(T,T')$ taxa \cite{AllenSteel2001}. This result was used to prove that the computation of $d_\TBR$ is fixed-parameter tractable (FPT); see \cite{Cygan:2015:PA:2815661} for an introduction to FPT.

Here we strengthen the bound given by Allen and Steel. We show that the reduced trees obtained from their kernelization algorithm in fact have at most $15 \cdot d_\TBR - 9$ taxa, and that this result is ``best possible'': for every $k\geq 2$, we demonstrate two trees with TBR distance $k$ such that, after the kernelization procedure has terminated, they have \emph{exactly} $15 \cdot d_\TBR - 9$ taxa. This proves that, if smaller kernels are to be obtained, additional reduction rules will be required. One natural candidate for a third reduction rule is the \emph{common cluster} reduction rule \cite{bordewich2017fixed}. However, using a slightly modified construction, we show that this rule (when applied in addition to the subtree and chain reduction rules) cannot improve upon the $15 \cdot d_\TBR - 9$ bound. 

A novel feature of our proofs is that they leverage recent insights from the phylogenetic \emph{networks} literature, where networks are essentially the generalization of phylogenetic trees to graphs \cite{HusonRuppScornavacca10}. Specifically, it was recently shown that if one embeds two trees $T$ and $T'$ into an unrooted phylogenetic network $N =(V,E)$, then the minimum value of $|E|-(|V|-1)$ ranging over all such $N$ will be equal to $d_\TBR(T,T')$ \cite{van2018unrooted}. This is significant because it attaches a static, graph-based interpretation to \blue{the} TBR distance: it allows us to view its computation as a graph/network-construction problem. In turn, this allows us to define and use unrooted analogues of \emph{generators} (i.e. backbone topologies) \cite{kelk2014constructing,lev2TCBB} which have been used extensively \purple{in the FPT literature on rooted phylogenetic networks} (see e.g. \cite{van2016hybridization} and references therein). Once viewed this way, the strengthened $15 \cdot d_\TBR - 9$ bound can be derived via a fairly straightforward counting argument. The generators also turn out to be invaluable in proving the tightness of the bound. 

As a spin-off to these results we show that the earlier-identified upper bound of $9k - 2$ \cite{approximationHN} on the size of the standard \polish{hybridization number} \steven{weighted} kernel \cite{sempbordfpt2007} \polish{for rooted trees} is also tight, and that in this case the cluster reduction can only improve the bound slightly, to $9k-4$, which as we demonstrate is also tight.

In the final part of the article we devote a discussion section to summarizing the (new) state of the algorithmic landsdcape for TBR distance and reflect upon the broader consequences of our strengthened bound on the size of the TBR kernel. We also list a number of \polish{related} phylogenetics problems where there is still quite some potential for improving bounds on kernel sizes.

\section{Preliminaries}\label{sec:prelim}

\noindent{\bf Unrooted phylogenetic trees and networks.} Throughout this paper $X$ denotes a finite set \polish{(of \emph{taxa})} with at least two elements. An {\it unrooted binary phylogenetic network} on $X$  is a  simple, connected, and undirected graph $N$ with $|X|$ vertices, called {\it leaves}, of degree one and bijectively labeled with $X$, and all other vertices of degree 3.  We define the {\it reticulation number} of $N$ as $r(N) = |E|-(|V|-1)$, where $E$ and $V$ are the edge and vertex sets of $N$, respectively. If $r(N)=0$, then $N$ is called an {\it unrooted binary phylogenetic tree} on $X$.  \\

\noindent{\bf Subtrees, chains, and clusters.} Let $N$ be an unrooted binary phylogenetic network on $X$. A {\it pendant subtree} of $N$ is an unrooted binary phylogenetic tree on a proper subset of $X$ that can be obtained from $N$ by deleting a single edge. For $n\geq 1$, let $C=(\ell_1,\ell_2,\ldots,\ell_n)$ be a sequence of distinct leaves in $X$ and, for each $i\in\{1,2,\ldots,n\}$, let $p_i$ denote the unique neighbor of $\ell_i$ in $N$. We call $C$ an $n$-chain of $N$ if there exists a path $p_1,p_2,\ldots,p_n$ in $N$ such that $p_2, \ldots, p_{n-1}$ is a simple path. That is, we optionally allow that $p_1 = p_2$ (i.e. $\ell_1$ and $\ell_2$ have a common parent) and/or $p_{n-1} = p_{n}$ (i.e. $\ell_{n-1}$ and $\ell_{n}$ have a common parent). Furthermore, $n$ is referred to as the {\it length} of $C$. By definition, note that each element in $X$ is a chain of length 1 in $N$. Lastly, for $Y\subset X$, we say that $Y$ is a {\it cluster} of $N$ if there exists a single edge in $N$ whose deletion disconnects $N$ into two parts such that the leaves of one part are bijectively labeled by elements in $Y$ while the leaves of the other part are bijectively labeled by elements in $X-Y$. If $|Y|$=1, then the cluster is called {\it trivial} \blue{and, otherwise, it is called {\it non-trivial}}. Note that, if $Y$ is a cluster of $N$, then $X-Y$ is also a cluster of $N$. \polish{We say that $Y, X-Y$ is a {\it bipartition} of $N$.} \\

\noindent{\bf Generators.} Rooted generators have played an important role in establishing kernelization results for problems on rooted trees~\cite{approximationHN,vanIersel20161075,ierselLinz2013}, but have only received very little attention~\cite{GBP2012} as a tool to tackle problems on unrooted trees. Here we give a definition of an unrooted generator that we will subsequently use to establish an improved kernel for the problem of computing the TBR distance (formally defined below) between two trees. Let $k$ be a positive integer.
For $k\geq 2$, a {\it $k$-generator} (or short {\it generator} when $k$ is clear from the context) is a connected cubic multigraph with edge set $E$ and vertex set $V$ such that $k=|E|-(|V|-1)$. Furthermore, for $k=1$, we define the graph that consists of a single vertex $u$ and a loop edge $\{u,u\}$ to be the unique {\it $1$-generator}. The edges of a generator are also called its {\it sides}. Intuitively, the sides are the places where leaves can be attached in order to obtain an unrooted binary phylogenetic network from a generator. We now formalize this concept. Let $G$ be a $k$-generator, let $\{u,v\}$ be a side of $G$, and let $Y$ be a set of leaves. The operation of subdividing $\{u,v\}$ with $|Y|$ new vertices and, for each such new vertex $w$, adding a new edge $\{w,\ell\}$, where $\ell\in Y$ such that $Y$ bijectively labels the new leaves is referred to as {\it attaching} $Y$  to  $\{u,v\}$. \blue{Additionally, if $G$ is the 1-generator, then the degree-2 vertex $u$ is suppressed after attaching $Y$ to $\{u,u\}$. Moreover,} if \purple{at least two} new leaves are attached to $G$ in a way that at least one new leaf is attached to each loop and to each pair of parallel edges, then the resulting graph is an unrooted binary phylogenetic network $N$ with $r(N)=k$. Note that $N$ has no pendant subtree with more than a single leaf. Conversely, we obtain $G$ from $N$ by deleting all leaves and, repeatedly, suppressing any resulting  degree-2 vertices. We say that $G$ {\it underlies} $N$. In summary, we make the following observation.

\begin{observation}\label{o:generator}
Let $N$ be an unrooted binary phylogenetic network with  \polish{$r(N)=k\geq 2$} and with no pendant subtree of size at least two.
Then the graph $G$ that is obtained from $N$ by  deleting all leaves and, repeatedly, suppressing any resulting  degree-2 vertices is a $k$-generator. Conversely, we obtain $N$ from $G$ by attaching to each side $s=\{u,v\}$ of $G$ a (possibly empty) set of leaves.
\end{observation}

\noindent{\bf Tree bisection and reconnection.} Let $T$ be an unrooted binary phylogenetic tree on X. Apply the following three-step operation to $T$:
\begin{enumerate}
\item Delete an edge in $T$ and suppress any resulting degree-2 vertex so that two new \steven{unrooted} binary phylogenetic trees, say $T_1$ and $T_2$, are obtained.
\item If $T_1$ (resp. $T_2$) has at least one edge, subdivide an edge in $T_1$ (resp. $T_2$) with a new vertex $v_1$ (resp. $v_2$) and otherwise set $v_1$ (resp. $v_2$) to be the single isolated vertex of $T_1$ (resp. $T_2$).
\item Add a new edge \polish{$\{v_1,v_2\}$} to obtain a new unrooted binary phylogenetic tree $T'$ on $X$.
\end{enumerate}
We say that $T'$ has been obtained from $T$ by a single  {\it tree bisection and reconnection} (TBR) operation. Furthermore, we define the TBR {\it distance}  between two  unrooted binary phylogenetic trees $T$ and $T'$ on $X$,  denoted by $d_\TBR(T,T')$, to be the minimum number of TBR operations that is required to transform $T$ into $T'$. It is well known that, for any such pair of trees, one can always obtain one from the other by a sequence of TBR operations. In particular, $d_\TBR$ is a metric~\cite{AllenSteel2001}. By building on an earlier result by Hein et al.~\cite[Theorem 8]{hein1996complexity}, Allen and Steel~\cite{AllenSteel2001} established NP-hardness of computing the TBR distance.\\

\noindent{\bf \blue{Unrooted} minimum hybridization.} In ~\cite{van2018unrooted}, it was shown that computing the TBR distance for a pair of unrooted binary phylogenetic trees $T$ and $T'$ is equivalent to a problem that is concerned with computing the minimum number of \polish{extra edges} required to simultaneously explain $T$ and $T'$. To  describe this problem precisely, let $N$ be an unrooted binary phylogenetic network on $X$, and let $T$ be an unrooted binary phylogenetic tree on $X$. We say that $N$ {\it displays} $T$, if  $T$ can be obtained from a subtree of $N$ by suppressing degree-2 vertices. Furthermore, for two unrooted binary phylogenetic trees $T$ and $T'$ on $X$, we set
$$h^u(T, T') = \min_ N\{r(N)\},$$ where the minimum is taken over all unrooted binary phylogenetic networks  on  $X$ that display $T$ and $T'$. \polish{The value $h^u(T, T')$ is known as the \emph{hybridization number} of $T$ and $T'$ \cite{van2018unrooted}.} \\

\noindent{\sc Unrooted-Hybridization-Number (UHN)}\\
\noindent{\bf Input.} Two unrooted binary phylogenetic trees $T$ and $T'$ on $X$.\\
\noindent{\bf Output.} An unrooted binary phylogenetic network $N$ on $X$ that displays $T$ and $T'$ and such that $r(N)=h^u(T,T')$.\\

\noindent We are now in a position to formally state the aforementioned equivalence that was established in~\cite[Theorem 3]{van2018unrooted}.

\begin{theorem}\label{t:tbr-equiv}
Let $T$ and  $T'$ be two unrooted binary phylogenetic trees on $X$. Then $$d_\TBR(T,T')=h^u(T,T').$$
\end{theorem}


\noindent {\bf Reductions and kernelization.}
While computing the TBR distance is NP-hard, it was also shown in~\cite{AllenSteel2001} that the problem is fixed-parameter tractable when parameterized by $d_\TBR$. For two unrooted binary phylogenetic trees $T$ and $T'$  on $X$, the authors used the following two reductions to kernelize the problem.\\

{\bf Subtree reduction.} Replace a maximal pendant subtree with at least two leaves that is common to $T$ and $T'$ by a single leaf with a new label.\\

{\bf Chain reduction.} Replace a maximal $n$-chain with $n\geq 4$ that is common to $T$ and $T'$ by a 3-chain with three new leaf labels correctly oriented to preserve the direction of the chain. For an illustration of this reduction, see Figure~\ref{fig:chain-red}.\\

\begin{figure}[!ht]
\center
\includegraphics[width=\textwidth]{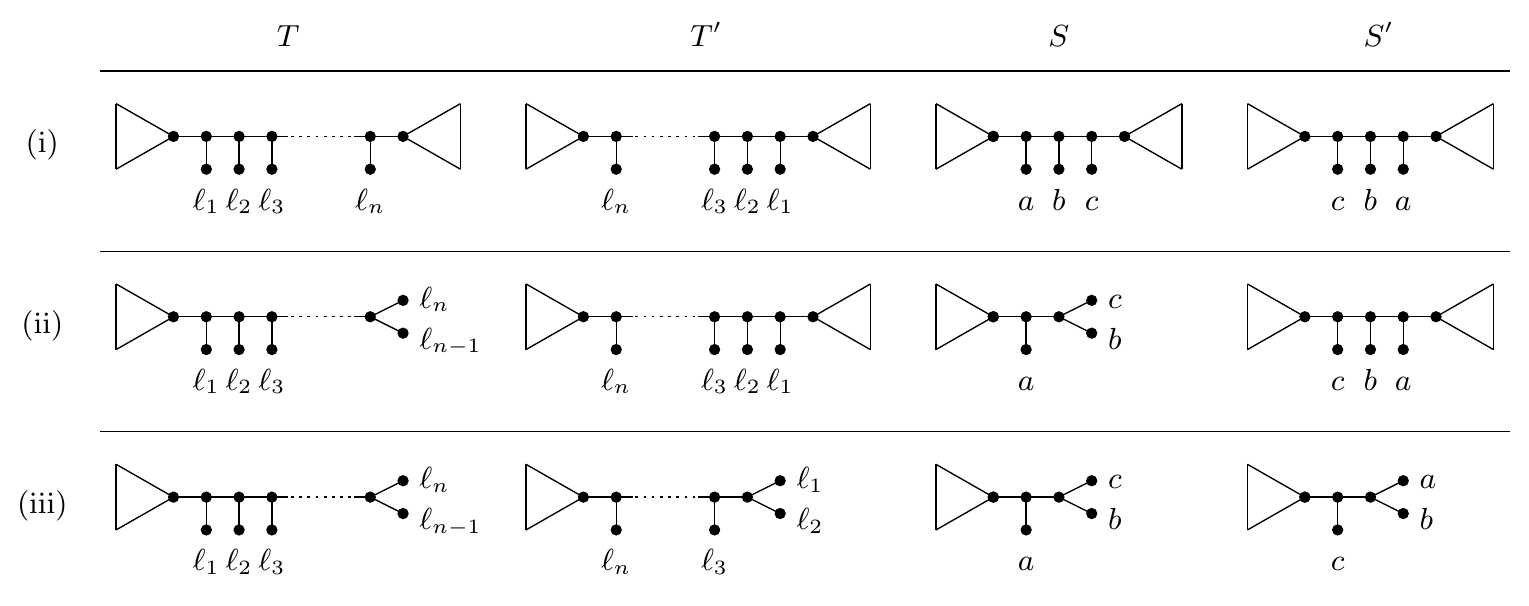}
\caption{A chain reduction applied to the common $n$-chain $C=(\ell_1,\ell_2,\ldots,\ell_n)$ with $n\geq 4$ of two unrooted binary phylogenetic trees $T$ and $T'$. In the reduced trees $S$ and $S'$, $C$ is replaced with the 3-chain $(a,b,c)$. Triangles indicate subtrees of $T$, $T'$, $S$, and $S'$. As shown in (ii) and (iii), $C$ may be pendant in one or both of $T$ and $T'$ and, hence, $\ell_1$ and $\ell_2$ have the same parent and/or $\ell_{n-1}$ and $\ell_{n}$ have the same parent. Note that, if $T$ and $T'$ both have the property that $\ell_1$ and $\ell_2$ have the same parent, then $C$ is in fact a common pendant subtree and reduced under the subtree reduction.}
\label{fig:chain-red}
\end{figure}

\noindent If $T$ and $T'$ cannot be reduced under the subtree (resp. chain) reduction, we say that $T$ and $T'$ are {\it subtree} (resp. {\it chain}) {\it reduced}.\\

The next theorem, which is due to~\cite[Theorem 3.4]{AllenSteel2001}, shows that both reductions are safe, i.e. they do not change the TBR distance.

\begin{theorem}\label{t:safe}
Let $T$ and $T'$ be two  unrooted binary phylogenetic trees on $X$, and let $S$ and $S'$ be two trees obtained from $T'$ and $T'$, respectively, by applying a single subtree or chain reduction. Then $d_\TBR(T,T')=d_\TBR(S,S')$.
\end{theorem}

Repeated applications of the previous lemma allow us to obtain two trees, say $S$ and $S'$ on $X'$, from $T$ and $T'$, \purple{respectively}, that are subtree and chain reduced such that $d_\TBR(T,T')=d_\TBR(S,S')$. 
The importance of $S$ and $S'$ lies in the following kernelization result that we have alluded to in the introduction and that is a direct consequence of~\cite[Lemmas 3.6 and 3.7]{AllenSteel2001}.

\begin{theorem}\label{t:tbr-first-kernel}
Let $T$ and $T'$ be two unrooted binary phylogenetic trees on $X$, and let $S$ and $S'$ be a subtree and chain reduced tree pair on $X'$ that has been obtained from $T$ and $T'$ by repeated applications of the subtree and chain reduction. Then \blue{$|X'|\leq 28d_\TBR(T,T').$}
\end{theorem}

Noting that the subtree and chain reduction can be applied in time that is polynomial in the size of $X$, it immediately follows that computing the TBR distance is fixed-parameter tractable when parameterized by this minimum distance.\\

\noindent
\steven{{\bf Cluster reduction.} In \cite{bordewich2017fixed} it is shown that if $T$ and $T'$ have a common non-trivial cluster, computation of $d_\TBR(T,T')$ can be reduced to \purple{the} computation of $d_\TBR$ on two new smaller pairs of trees obtained by
\polish{decomposing} $T$ and $T'$ around the common cluster. 
(Compared to the subtree and chain reductions \polish{the} cluster reduction involves a number of subtleties; for brevity we do not describe them here.) This decomposition procedure can be continued until trees are obtained without non-trivial common clusters.} \blue{If $T$ and $T'$ do not have a non-trivial common cluster, we say that $T$ and $T'$ are {\it cluster reduced}.}\\  
\\
\noindent{\bf Maximum Parsimony distance.} A {\it character} $f$ on $X$ is a function $f:X \rightarrow C$, where $C=\{c_1,c_2,\ldots,c_r\}$ is a set of {\it character states} for some positive integer $r$. Let $T$ be an unrooted phylogenetic tree on $X$ with vertex set $V$, and let $f$ be a character on $X$ whose set of character states is $C$. An {\it extension} $g$ of $f$ to $V$ is  a function $g: V \rightarrow C$ such that \purple{$g(\ell)=f(\ell)$} for each $\ell\in X$. Given an extension $g$ of $f$, let $l_g(T)$ denote the number of edges $\{u,v\}$ in $T$ such that $g(u)\ne g(v)$. Then the {\it parsimony score} of $f$ on $T$, denoted by $l_f(T)$, is obtained by minimizing  $l_g(T)$ over all possible extensions $g$ of $f$. Lastly, for two unrooted phylogenetic trees $T$ and $T'$ on $X$, we define the {\it maximum parsimony distance} $d_{\MP}$ as  $$d_{\MP}(T,T')=\max_f |l_f(T)-l_f(T')|.$$ Importantly, for two unrooted binary phylogenetic trees, the maximum parsimony distance is a lower bound on the TBR distance as noted in the discussion of~\cite{fischer2014}. \blue{We will freely use this fact throughout the rest of this paper. }For more details on the maximum parsimony distance, we refer the interested reader to~\cite{fischer2014,kelk2016reduction,moulton2015}. 

\section{Tight kernels for computing the TBR distance}
\purple{In this section, we use generators to reanalyze the TBR distance kernelization result by Allen and Steel~\cite{AllenSteel2001} and establish a kernel for the problem whose size is significantly smaller.} Let $T$ and $T'$ be two unrooted binary phylogenetic trees on $X$, and let $N$ be an unrooted phylogenetic network on $X$ that displays $T$ and $T'$. Furthermore, let $C=(\ell_1,\ell_2,\ldots,\ell_n)$ be an $n$-chain of $N$, and for each $i\in\{1,2,\ldots,n\}$, let $p_i$ be the unique neighbor of $\ell_i$ in $N$. Since each embedding of $T$ (resp. $T'$) into $N$ uses either all or all but one edge in $\{\{p_1,p_2\},\{p_2,p_3\},\ldots,\{p_{n-1},p_n\}\}$, it is straightforward to check that exactly one of the following three \steven{cases} holds.
\begin{enumerate}
\item $C$ is a chain of $T$ and $T'$
\item There exists a \steven{\emph{breakpoint}} $i\in\{1,2,\ldots,n-1\}$, such that $$C_1=(\ell_1,\ell_2,\ldots,\ell_i)\textnormal{ and } C_2=(\ell_{i+1},\ell_{i+2},\ldots,\ell_{n})$$ are chains of $T$, and $C$ is a chain of $T'$, or $C$ is a chain of $T$ and $C_1$ and $C_2$ are chains of $T'$. 
\item There exist two \blue{not necessarily distinct} breakpoints $i,j\in\{1,2,\ldots,n-1\}$ such that $$(\ell_1,\ell_2,\ldots,\ell_i)\textnormal{ and }  (\ell_{i+1},\ell_{i+2},\ldots,\ell_n)$$ are chains of $T$ and $$(\ell_1,\ell_2,\ldots,\ell_j)\textnormal{ and }  (\ell_{j+1},\ell_{j+2},\ldots,\ell_n)$$ are chains of $T'$.
\end{enumerate} 
We say that $C$ has 0, 1, or 2 {\it breakpoints} relative to $T$ and $T'$ depending on whether $C$ is not cut (Case (1)), cut once (Case (2)), or cut twice (Case (3)). Intuitively, the number of breakpoints indicates how many trees of $T$ and $T'$ do not have $C$ as a chain.\\


\begin{lemma}\label{l:edges}
Let $N$ be an unrooted binary phylogenetic network with  $r(N)=k\geq 2$ \blue{and with no pendant subtree of size at least two}. Furthermore, let $G$ be the generator that underlies $N$. Then $G$ has $3(k-1)$ sides.
\end{lemma}

\begin{proof}
Let $E$ be the edge set of $G$, and let $V$ be the vertex set of $G$. By construction of $G$ from $N$, recall that each vertex of $G$ has degree 3 and that $|E|-|V|+1=k$. Hence, 
$$2|E|=3|V|=3(|E|-k+1),$$ where the first equality is due to the Handshaking Lemma. Now, solving for $|E|$, we have $|E|=3(k-1)$. Since $E$ is equal to the set of sides of $G$, the lemma follows.\qed
\end{proof}


\begin{lemma}\label{l:chain-length}
Let $S$ and $S'$ be two unrooted binary phylogenetic trees $X$, and let $N$ be an unrooted phylogenetic network on $X$ that displays $S$ and $S'$. Furthermore, let $C$ be an $n$-chain of $N$. Depending on the number of breakpoints of $C$ relative to $S$ and $S'$, $n$ is bounded from above in the following way.
\begin{enumerate}
\item Suppose that $S$ and $S'$ are subtree and chain reduced. Then $n\leq 3$ if $C$ has 0 breakpoints, $n\leq 6$ if $C$ has 1 breakpoint, and $n\leq 9$ if $C$ has 2 breakpoints.
\item \blue{Suppose that $S$ and $S'$ are subtree, chain, and cluster reduced. Then $n\leq 3$ if $C$ has 0 breakpoints, $n\leq 6$ if $C$ has 1 breakpoint, and $n\leq 7$ if $C$ has 2 breakpoints.}
\end{enumerate}
\end{lemma}

\begin{proof}
\blue{We establish the lemma for when $C$ has two breakpoints. The other cases are similar and omitted.} Let $C=(\ell_1,\ell_2,\ldots,\ell_n)$, and let $i$ and $j$ be the two breakpoints of $C$ relative to $S$ and $S'$. Without loss of generality, we may assume that $i\leq j$.  Since $C$ has two breakpoints, one of the following two cases applies.
\begin{enumerate}[(a)]
\item If $i=j$, then the chains $(\ell_1,\ell_2,\ldots,\ell_i)\textnormal{ and }  (\ell_{i+1},\ell_{i+2},\ldots,\ell_n)$ are common to $S$ and $S'$. 
\item If $i<j$, then the chains $(\ell_1,\ell_2,\ldots,\ell_i),  (\ell_{i+1},\ell_{i+2},\ldots,\ell_j),\textnormal{ and } (\ell_{j+1},\ell_{j+2},\ldots,\ell_{n})$ are common to $S$ and $S'$.
\end{enumerate}
\blue{First, suppose that $S$ and $S'$ are subtree and chain reduced and, towards a contradiction, assume that $n>9$.
Regardless of whether (a) or (b) applies, it follows by the pigeonhole principle that one of the smaller chains has length at least 4 and is common to $S$ and $S'$; thereby contradicting that $S$ and $S'$ are chain reduced. 
Second, suppose that $S$ and $S'$ are subtree, chain, and cluster reduced and, towards a contradiction, assume that $n>7$. If $i<j$, observe that $\{\ell_{i+1},\ell_{i+2},\ldots,\ell_j\}$ is the leaf set of a pendant subtree in $S$ and $S'$. Hence $\{\ell_{i+1},\ell_{i+2},\ldots,\ell_j\}$ is a cluster that is common to $S$ and $S'$. Since $S$ and $S'$ are cluster reduced, this implies that $i+1=j$. Now, regardless of whether (a) or (b) applies, it follows again by the pigeonhole principle that one of the other chains that is common to $S$ and $S'$ has length at least 4; a contradiction.}
\qed
\end{proof}

Let $T$ and $T'$ be two unrooted binary phylogenetic trees. Using generators, the next lemma exploits the equivalence between computing the TBR distance and UHN to establish a new and improved upper bound on the number of leaves of a pair of \purple{subtree and chain reduced} trees.

\begin{lemma}\label{l:tbr-cs}
Let $S$ and $S'$ be two unrooted binary phylogenetic trees on $X$ that are subtree and chain reduced, \polish{and $d_\TBR(S,S')\geq 2$}. Then, $|X|\leq 15d_\TBR(S,S')-9$.
\end{lemma}

\begin{proof}
Let $N$ be an unrooted binary phylogenetic network on $X$ with edge set $E$ \blue{and vertex set $V$} that displays $S$ and $S'$ such that $$r(N)=h^u(S,S')=d_\TBR(S,S')\steven{=}k\geq 2.$$ By Theorem~\ref{t:tbr-equiv}, such a network exists.
Let $G$ be the generator that underlies $N$. \blue{Furthermore, let $D$ and $D'$ be two subdivisions of $S$ and $S'$, respectively, in $N$. Since $N$ displays $S$ and $S'$, such subdivisions exist. If $D$ is a spanning tree of $N$, then set $B=D$ and, otherwise, set $B$ to be a spanning tree of $N$ obtained from $D$ by greedily adding edges. Similarly,  if $D'$ is a spanning tree of $N$, then set $B'=D'$ and, otherwise, set $B'$ to be a spanning tree of $N$ obtained from $D'$ by greedily adding edges. (Observe that $B$ and $B'$ may have unlabeled leaves.) Moreover, as $ |E|=|V|-1+k$, observe that each of $B$ and $B'$ has $|V|-1=|E|-k$ edges, where the left-hand side of the equation follows from the definition of a spanning tree.}


\blue{Now, as $S$ and $S'$ are subtree reduced, it is sufficient to attach leaves to the sides of $G$ in the process of obtaining $N$ from $G$ (see Observation~\ref{o:generator}).} Let $s=\{u,v\}$ be a side of $G$.  We next assign a {\it cut count} $c_s$ to $s$.  
First, if $s$ is decorated with at least one leaf in obtaining $N$ from $G$, let $Y=\{\ell_1,\ell_2,\ldots,\ell_n\}$ be the set of leaves that are attached to $s$. By construction, $N$ has a maximal $n$-chain $C$ whose leaves are bijectively labeled with the elements in $Y$. Let $P$ be the path from $u$ to $v$ in $N$ that has length $n+1$ and whose vertices (except for $u$ and $v$) are neighbors of the elements in $Y$. We  define $c_s$ to be the number of \blue{trees} in $\{B,B'\}$ that do not use all edges of $P$. Since $B$ and $B'$ \blue{both span $N$}, note that there is at most one edge that is not used by $B$ and at most one (not necessarily distinct) edge that is not used by $B'$. Hence, relative to $C$, we have $b_C\leq c_s$, where $b_C$ is the number of breakpoints of $C$ relative to $S$ and $S'$. \blue{It follows from Lemma~\ref{l:chain-length}.1 that}, if $c_s=0$, then at most 3 taxa can be attached to $s$ in obtaining $N$ from $G$. Similarly, if $c_s=1$, then at most 6 taxa can be attached to $s$ and, if $c_s=2$, then at most \polish{9} taxa can be attached to $s$ in this process. Second, if $s$ is not decorated with any leaf when obtaining $N$ from $G$, then $\{u,v\}$ is an edge in $N$, and we define $c_s$ to be the number of \blue{trees} in $\{B,B'\}$ that do not use $\{u,v\}$. \blue{Since we are establishing an upper bound on $|X|$, we may assume for the remainder of this proof that the bounds established in Lemma~\ref{l:chain-length} also apply for when $s$ is not \purple{decorated}.} By assigning a cut count to each side of $G$, \polish{and recalling that both $B$ and $B'$ contain $|E|-k$ edges,} it is easily checked that  
\blue{$$\sum_s c_s=2k,$$ where the sum is taken over all sides of $G$. }

Now, let $0\leq q\leq k$ be the number of sides of $G$ whose cut count is equal to two. Consequently, \polish{by leveraging the above equality}, there are $2(k-q)$ sides whose cut count is equal to 1 \blue{and, by Lemma~\ref{l:edges},  $3(k-1) - (k + k -q)$ sides whose cut count is zero.  Hence, as $S$ and $S'$ are subtree and chain reduced, this implies that}
$$|X|\leq 9q + 6\cdot 2(k-q) + 3(3(k-1) - (k + k -q))=15k-9=15d_\TBR(S,S')-9,$$
thereby establishing the lemma.\qed
\end{proof}
\blue{\noindent {\bf Remark.} Let $S$ and $S'$ be two unrooted binary phylogenetic trees on $X$ that are subtree, chain, and cluster reduced. Clearly,  Lemma~\ref{l:tbr-cs} holds for $S$ and $S'$. Alternatively, by using an argument that is similar to that in the proof of Lemma~\ref{l:tbr-cs} as well as the bounds derived in Lemma~\ref{l:chain-length}.2, we obtain $$|X|\leq 7q + 6\cdot 2(k-q) + 3(3(k-1) - (k + k -q))=-2q+15k-9.$$ Now noting that the right-hand side of the last inequality is maximized for $q=0$, we again have $$|X|\leq -2q+15k-9\leq 15k-9=15d_\TBR(S,S')-9.$$ For $k\geq 4$, we will see in the next section that  $15k-9$ is a tight upper bound on the size of the leaf set $X$ of two unrooted binary phylogenetic trees $S$ and $S'$ with $d_\TBR(S,S')=k$ and that are reduced under all three reductions. To establish this result, the fact that no side has a cut count of two (i.e. $q=0$)  gives us some important clues about the properties of $S$ and $S'$. In particular, it turns out that  $S$ and $S'$ are displayed by an unrooted binary phylogenetic network $N$ on $X$ with $r(N)=k$ such that $N$ has no chain of length 7. For full details, see Section~\ref{sec:tbr-ex}.}\\

In comparison to Theorem~\ref{t:tbr-first-kernel} \blue{that was first established in~\cite{AllenSteel2001}}, the next theorem, which is an immediate consequence of Theorem~\ref{t:safe} and Lemma~\ref{l:tbr-cs}, establishes a significantly improved linear kernel  for computing the TBR distance.

\begin{theorem}\label{t:tbr-kernel}
Let $T$ and $T'$ be two unrooted binary phylogenetic trees on $X$, \polish{where $d_\TBR(T,T')\geq 2$}, and let $S$ and $S'$ be a subtree and chain reduced tree pair on $X'$ that has been obtained from $T$ and $T'$ by repeated applications of the subtree and chain reduction. Then $|X'|\leq 15d_\TBR(T,T')-9$. 
\end{theorem}


\section{\blue{Tight examples}}\label{sec:tbr-ex}
In this section, we show that the upper bounds on the size of \blue{the} leaf set of two unrooted binary phylogenetic trees that do not contain any common subtree and chain (and cluster) as established in Lemma~\ref{l:tbr-cs} are tight. To this end, we provide two \polish{families} of \blue{constructions}. 
Throughout this section, \purple{we  attach} leaves to a side $s$ of a generator that is depicted in \polish{Figure~\ref{fig:TBR-SC-ex} or Figure~\ref{fig:TBR-SCC-ex}}.
 If $s$ connects two vertices of a generator that lie on a horizontal line, we attach leaves to $s$ from left to right. Otherwise, we attach leaves to $s$ from top to bottom. \blue{Furthermore}, for a set of leaves that is attached to $s$, we order the elements from small to large and then attach from left to right or top to bottom in such a way that preserves the ordering of the elements.

We now begin with the description
\polish{of a family of pairs of} unrooted binary phylogenetic trees that are subtree and chain reduced (but not cluster reduced). For $k\geq 2$, consider the two unrooted binary phylogenetic trees $S_k$ and $S_k'$ with leaf set $X_k$ and \blue{$|X_k|=15k-9$} that are shown in Figure~\ref{fig:TBR-SC-ex}. 
It is easy to check that $S_k$ and $S_k'$ do not contain any common subtree of size at least 2 or any common $n$-chain with $n\geq 4$. Note however that $S_k$ and $S_k'$ do contain $k$ common clusters of size three. \blue{In particular, $S_k$ and $S_k'$ contain the common cluster $\{15i-8,15i-7,15i-6\}$ for each $i\in\{1,2,\ldots,k-1\}$ and, additionally, the common cluster $\{15k-14,15k-13,15k-12\}$.} 

\begin{figure}[t]
\center
\includegraphics[width=15cm]{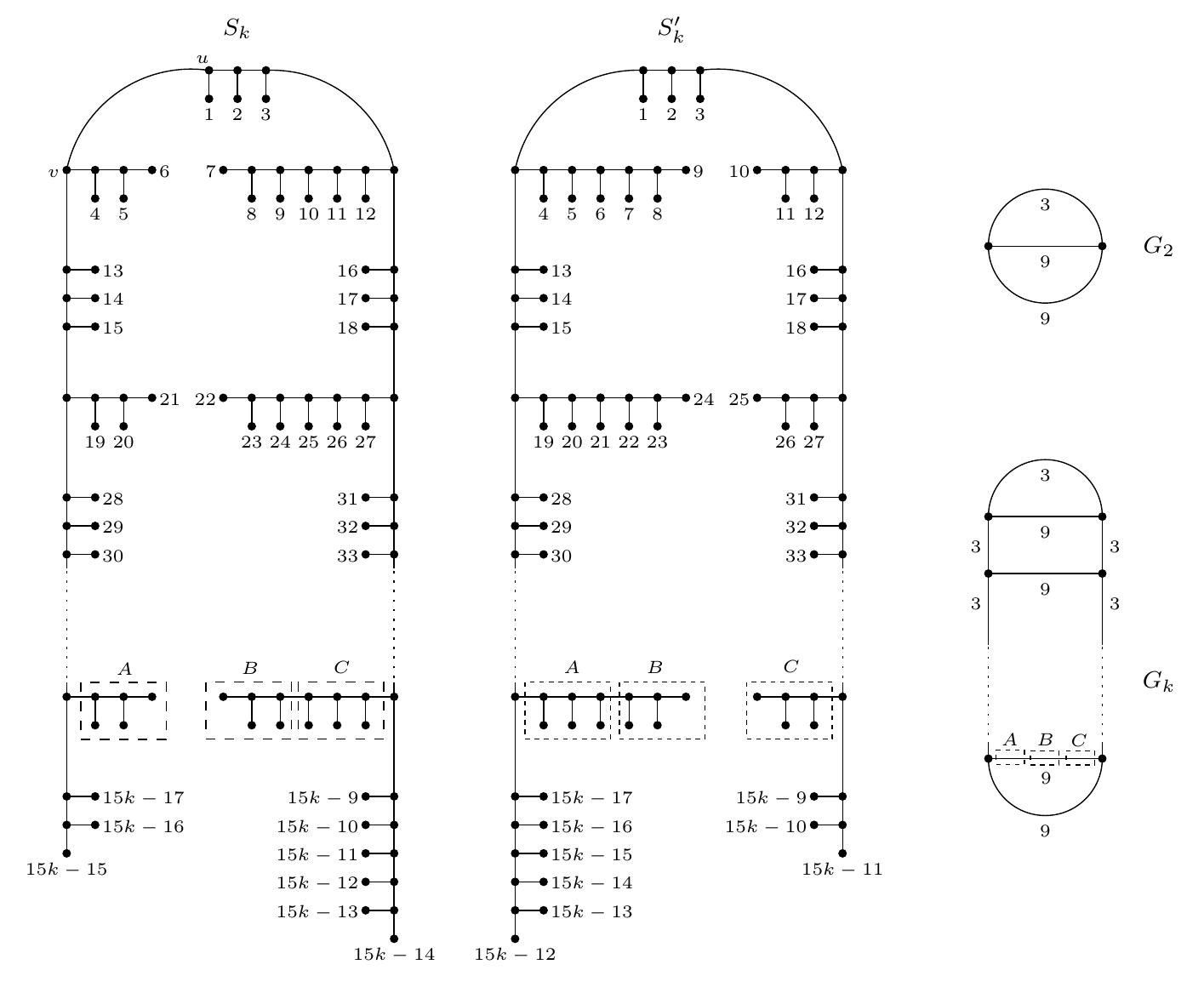}
\caption{\blue{The two unrooted binary phylogenetic trees $S_k$ and $S_k'$ as well as the generator $G_k$ (and $G_2$) that are used to show that the upper bound given in Lemma~\ref{l:tbr-cs} is tight for a pair of subtree and chain reduced trees. Blocks $A$, $B$, and $C$ indicate three chains each of length 3. For example, for $k=4$, block $A$ contains taxa 34, 35, and 36, block $B$ contains taxa 37, 38, and 39, and block $C$ contains taxa 40, 41, and 42. The edge $\{u,v\}$ of $S_k$ is used to argue that $k= d_\TBR(S_k,S_k')$ in the proof of Lemma~\ref{l:TBRsc=k}.}}
\label{fig:TBR-SC-ex}
\end{figure}

\begin{lemma}\label{l:TBR-cs-construction}
For $k\geq 2$, let $S_k$ and $S_k'$ be the two unrooted binary phylogenetic trees that are shown in Figure~\ref{fig:TBR-SC-ex}, and let $G_k$ be the generator that is shown in the same figure. There exists an unrooted binary phylogenetic network $N_k$ with $r(N_k)=k$ that displays $S_k$ and $S_k'$ and whose underlying generator is $G_k$.
\end{lemma}

\begin{proof}
The proof is constructive, i.e. starting with $G_k$, we attach leaves to the sides of $G_k$ to obtain a phylogenetic network that displays $S_k$ and $S_k'$. The lemma then follows from Observation~\ref{o:generator}.  \blue{Obtain an unrooted binary  phylogenetic network $N_k$ in the following way.}
\begin{enumerate}
\item Attach $\{1,2,3\}$ to the top-most horizontal side of $G_k$. 
\item For each $i\in\{1,2,\ldots,k-2\}$ in increasing order, perform the following three steps.
\begin{enumerate}
\item Attach $\{15i-11,15i-10,\ldots,15i-3\}$ to the top-most horizontal side of $G_k$ which is still undecorated.
\item Attach $\{15i-2,15i-1,15i\}$ to the top-most left-vertical side of $G_k$ which is still undecorated.
\item Attach $\{15i+1,15i+2,15i+3\}$ to the top-most-right vertical side of $G_k$ which is still undecorated.
\end{enumerate}
\item Attach $\{15k-26,15k-25,\ldots, 15k-18\}$ to the \blue{second-to-lowest} horizontal side of $G_k$.
\item Attach $\{15k-17,15k-16,\ldots,15k-9\}$ to the \blue{lowest} horizontal side of $G_k$.
\end{enumerate}
Observe that \purple{the label} of each side $s$ of $G_k$ as depicted in Figure~\ref{fig:TBR-SC-ex} refers to the number of leaves that is attached to $s$. Furthermore, note that $N_k$ has $k$ chains of length 9. Regarding each such chain as a sequence of three blocks where each block contains three leaves of the chain, which is indicated by $A$, $B$, and $C$ in Figure ~\ref{fig:TBR-SC-ex}, $S_k$ can be obtained from $N_k$ by breaking each 9-chain between $A$ and $B$ and suppressing all resulting degree-2 vertices, and $S_k'$ can be obtained from $N_k$ by breaking each 9-chain between $B$ and $C$ and suppressing all resulting degree-2 vertices. Hence, $N_k$ displays $S_k$ and $S_k'$. Moreover, by construction, $G_k$ underlies $N_k$. Let $E_k$ and $V_k$ be the edge and vertex set of $G_k$, respectively. We complete the proof by noting that, as $|E_k|-|V_k|+1=k$, it again follows by construction that $r(N_k)=k$.\qed
\end{proof}

\begin{lemma}\label{l:TBRsc=k}
For $k\geq 2$, let $S_k$ and $S_k'$ be the two unrooted binary phylogenetic trees that are shown in Figure~\ref{fig:TBR-SC-ex}. Then $d_\TBR(S_k,S'_k)=k$.
\end{lemma}

\begin{proof}
Let $N_k$ \polish{be} the unrooted binary phylogenetic network whose construction is described in the proof of Lemma~\ref{l:TBR-cs-construction}. Then it follows from the same lemma and Theorem~\ref{t:tbr-equiv} that
\begin{equation}\label{eq:one}
d_\TBR(S_k,S_k')=h^u(S_k,S_k')\leq r(N_k)=k.
\end{equation}
 We complete the proof by showing that $d_\TBR(S_k,S_k')\geq k$. Let $X_k$ be the leaf set of $S_k$ and $S_k'$, and let $C=\{0,1\}$. Furthermore, let $f:X_k\rightarrow C$ be \polish{the} character on $X_k$
\polish{defined as follows.} 
 %
 \blue{Consider the bipartition of $X_k$ induced by removing the edge $\{u,v\}$ as labeled in Figure~\ref{fig:TBR-SC-ex}. Give one side of the partition state 0 and the other state 1. Clearly, $l_f(S_k)=1$. Moreover, by applying Fitch's algorithm~\cite{fitch1971}, we see that $l_f(S_k')=k+1$.} Hence, \purple{by definition of the maximum parsimony distance, we have}
\begin{equation} \label{eq:two}
k= |1-(k+1)|\leq d_\MP(S_k,S_k') \leq d_\TBR(S_k,S_k'). 
\end{equation}
Combining Inequalities~\ref{eq:one} and~\ref{eq:two} establishes the lemma.\qed
\end{proof}

We now turn to a \blue{construction} \purple{for two} unrooted binary phylogenetic trees \purple{that} are subtree, chain, and cluster reduced. For $k\geq 4$, consider the two trees $S_k$ and $S_k'$ on leaf set $X_k$ that are shown in Figure~\ref{fig:TBR-SCC-ex}. As with the trees depicted in Figure~\ref{fig:TBR-SC-ex}, note that $|X_k|=15k-9$. 

\begin{figure}[!ht]
\center
\includegraphics[width=15cm]{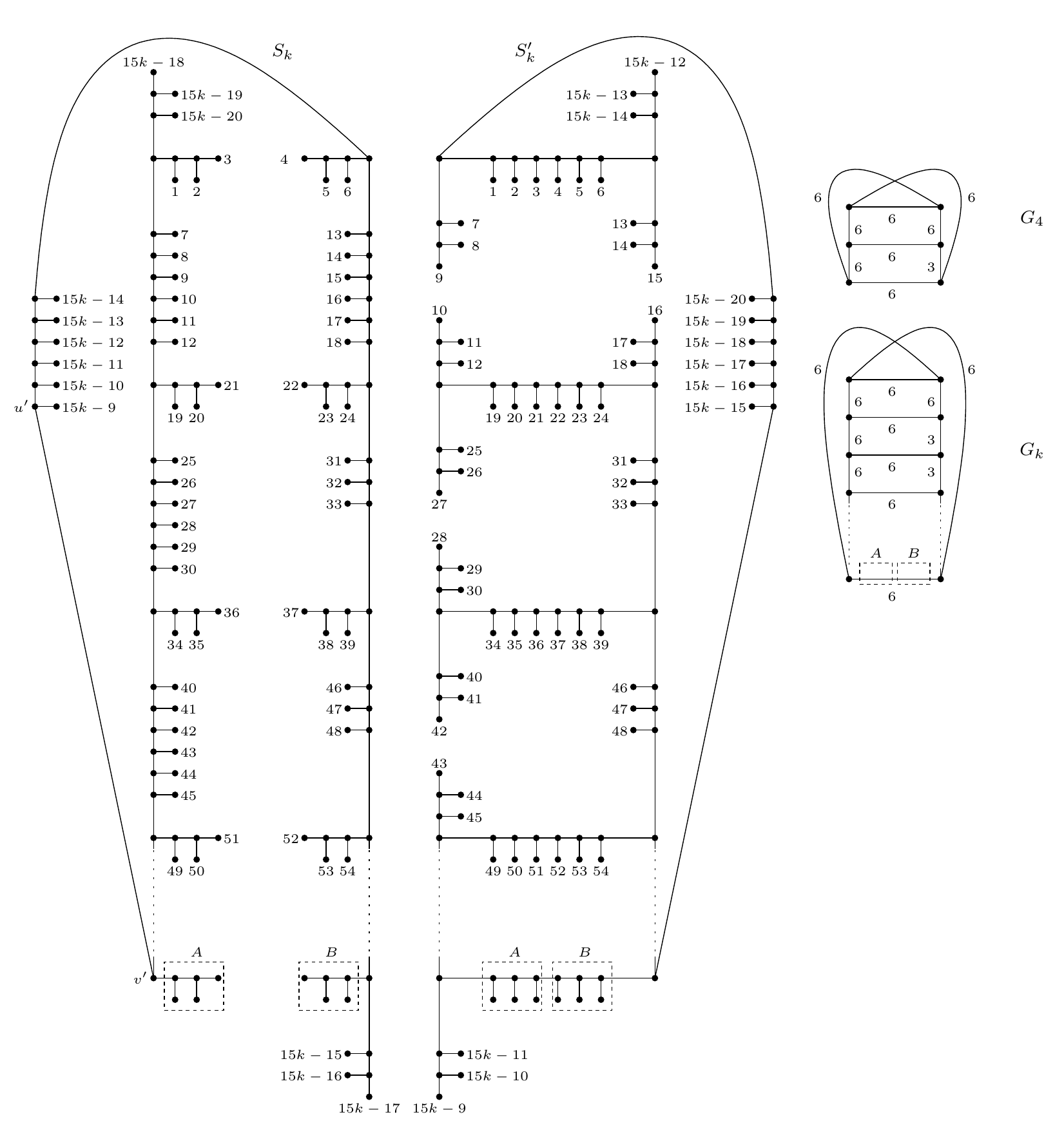}
\caption{The unrooted binary phylogenetic trees $S_k$ and $S_k'$ as well as the generator $G_k$ (and $G_4$) that are used to show that the upper bound given in Lemma~\ref{l:tbr-cs} is tight for a pair of subtree, chain, and cluster reduced trees. Blocks $A$ and $B$ indicate two chains each of length 3. For example, for $k=6$, block $A$ contains taxa 64, 65, and 66, and block $B$ contains taxa 67, 68, and 69. \blue{The edge $\{u',v'\}$ of $S_k$ is used to argue that $k= d_\TBR(S_k,S_k')$ in the proof of Lemma~\ref{l:TBRscc=k}.}}
\label{fig:TBR-SCC-ex}
\end{figure}

\begin{lemma}
For $k\geq 4$, let $S_k$ and $S_k'$ be the two unrooted binary phylogenetic trees that are shown in Figure~\ref{fig:TBR-SCC-ex}. Then $S_k$ and $S_k'$ are subtree, chain, and cluster reduced.
\end{lemma}

\begin{proof}
It is straightforward to check that $S_k$ and $S_k'$ are subtree reduced. To see that they are also chain-reduced, notice that $S_k'$ has $k$ maximal 6-chains. However, none of these chains (or a subchain of size at least 4) is also a chain of $S_k$. \blue{We complete the proof by showing that $S_k$ and $S_k'$ are also cluster reduced. Let $P$ be the path from  $15k-18$ to  $15k-17$ in $S_k$. Furthermore, let $e$ be an arbitrary edge of $S_k$, and let $Y_1,Y_2$ be the bipartition of $X_k$ such that one subtree obtained from $S_k$ by deleting $e$ has leaf set $Y_1$ while the other subtree has leaf set $Y_2$. It  suffices to show that, if $Y_1$ and $Y_2$ both have size at least two, then $Y_1,Y_2$ is not a bipartition of $S_k'$, i.e. there exists no edge in $S_k'$ whose removal results in two subtrees with leaf sets $Y_1$ and $Y_2$, respectively. First, suppose that $e$ is an edge that does not lie on $P$. Then for some $i\in\{1,2\}$, we have $|Y_i|\leq 3$. Moreover, if $|Y_i|\in\{2,3\}$, then it is easily verified that $Y_1,Y_2$ is not a bipartition of $S_k'$. Second, suppose, that $e$ lies on $P$. Then $|\{15k-18,15k-17\}\cap Y_1|=1$ and $|\{15k-18,15k-17\}\cap Y_2|=1$. Now,  let $f$ be an edge of $S_k'$ whose removal results in two subtrees that both have at least two leaves, the leaf set of one subtree contains $15k-18$ and the leaf set of the other subtree contains $15k-17$. Clearly, the unique edge \purple{$f$} in $S_k'$ that satisfies all three properties is the second edge on the path from $15k-18$ to  $15k-17$. Let $Z_1,Z_2$ be the bipartition of $X_k$ such that $Z_1$ and $Z_2$ are the leaf sets of the two trees resulting from the deletion of $f$ in $S_k'$. Without loss of generality, we may assume that $\{9,15\}\subset Z_1$ and $\{10,16\}\subset Z_2$. But $Z_1,Z_2$ is not a bipartition of $S_k$  because we can either separate 9 and 10, or 15 and 16 by removal of a single edge in $S_k$ but not both. It now follows that $S_k$ and $S_k'$ are cluster reduced. This completes the proof of the lemma.}\qed
\end{proof}

\begin{lemma}\label{l:TBR-scc-construction}
For $k\geq 4$, let $S_k$ and $S_k'$ be the two unrooted binary phylogenetic trees that are shown in Figure~\ref{fig:TBR-SCC-ex}, and let $G_k$ be the generator that is shown in the same figure. There exists an unrooted binary phylogenetic network $N_k$ with $r(N_k)=k$ that displays $S_k$ and $S_k'$ and whose underlying generator is $G_k$.
\end{lemma}

\begin{proof}
The proof is similar to that of Lemma~\ref{l:TBR-cs-construction}.  Starting with $G_k$, obtain an unrooted binary phylogenetic network $N_k$ in the following way.
\begin{enumerate}
\item Attach $\{1,2,\ldots,6\}$ to the top-most horizontal side of $G_k$. 
\item Attach $\{7,8,\ldots,12\}$ to the top-most left-vertical side of $G_k$.
\item Attach $\{13,14,\ldots,18\}$ to the top-most right-vertical side of $G_k$.
\item For each $i\in\{2,\ldots,k-2\}$ in increasing order, perform the following three steps.
\begin{enumerate}
\item Attach $\{15i-11,15i-10,\ldots,15i-6\}$ to the top-most horizontal side of $G_k$ which is still undecorated.
\item Attach $\{15i-5,15i-4,\ldots,15i\}$ to the top-most left-vertical side of $G_k$ which is still undecorated.
\item Attach $\{15i+1,15i+2,15i+3\}$ to the top-most right-vertical side of $G_k$ which is still undecorated.
\end{enumerate}
\item  Attach $\{15(k-2)+4,15(k-2)+5,\ldots,15(k-2)+9\}$ to the lowest horizontal side of $G_k$.
\item Attach $\{15(k-2)+10,15(k-2)+11,\ldots,15(k-2)+15\}$ to the undecorated curved side of $G_k$ that connects the top-left vertex with the bottom-right vertex.
\item Attach $\{15(k-2)+16,15(k-2)+17,\ldots,15(k-2)+21\}$ to the undecorated curved side of $G_k$ that connects the top-right vertex with the bottom-left vertex.
\end{enumerate}
With $15(k-2)+21=15k-9$, it follows that $N_k$ has $15k-9$ leaves. Observe that $N_k$ has $2k$ chains of length 6. Regard each such chain as a sequence of two blocks where each block contains three leaves of the chain, which is indicated by $A$ and $B$ in Figure ~\ref{fig:TBR-SCC-ex}. In the following, we refer to the process of deleting the unique edge of a 6-chain that has one vertex in block $A$ and one in block $B$ as {\it breaking a chain}. Now, $S_k$ can be obtained from $N_k$ by breaking each of the $k-1$ 6-chains whose leaves decorate the $k-1$ horizontal sides of $G_k$, breaking the 6-chain whose leaves decorate the curved side of $G_k$ that connects the top-left vertex with the bottom-right vertex, and suppressing all resulting degree-2 vertices. Similarly, $S_k'$ can be obtained from $N_k$ by breaking each of the $k-2$ 6-chains whose leaves decorate the $k-2$ left-vertical sides of $G_k$, breaking the 6-chain whose leaves decorate the top-most right-vertical side of $G_k$, breaking the 6-chain whose leaves decorate the curved side of $G_k$ that connects the top-right vertex with the bottom-left vertex, and suppressing all resulting degree-2 vertices. Hence, $N_k$ displays $S_k$ and $S_k'$. Moreover, by construction and Observation~\ref{o:generator}, $G_k$ underlies $N_k$. We complete the proof by noting that, as $|E_k|-|V_k|+1=k$, it again follows by construction that $r(N_k)=k$, where $V_k$ and $E_k$ is the vertex and edge set of $G_k$, respectively.\qed
\end{proof}

\begin{lemma}\label{l:TBRscc=k}
For $k\geq 4$, let $S_k$ and $S_k'$ be the two unrooted binary phylogenetic trees that are shown in Figure~\ref{fig:TBR-SCC-ex}. Then $d_\TBR(S_k,S'_k)=k$.
\end{lemma}

\begin{proof}
By considering the unrooted binary phylogenetic network $N_k$ that is described in the proof of Lemma~\ref{l:TBR-scc-construction} (instead of the one described in the proof of Lemma~\ref{l:TBR-cs-construction}), the proof of this lemma can be established in exactly the same way as the proof of Lemma~\ref{l:TBRsc=k}. \blue{Note that the edge $\{u',v'\}$ as depicted in $S_k$ of Figure~\ref{fig:TBR-SCC-ex} plays the role of the edge $\{u,v\}$ that is used in the proof of Lemma~\ref{l:TBRsc=k}.}\qed
\end{proof}

We are now in a position to establish the main result of this section.

\begin{theorem}\label{t:TBR-tight}
Let $S$ and $S'$ be two unrooted binary phylogenetic trees on $X$. If $S$ and $S'$ are subtree and chain reduced and $d_\TBR(S,S')\geq 2$, then $|X|\leq 15d_\TBR(S,S')-9$ is a tight bound. Moreover, if $S$ and $S'$ are subtree, chain, and cluster reduced and $d_\TBR(S,S')\geq 4$, then $|X|\leq 15d_\TBR(S,S')-9$ is again a tight bound. 
\end{theorem}

\begin{proof}
Suppose that $S$ and $S'$ are subtree and chain reduced. It immediately follows from Lemma~\ref{l:tbr-cs} that $|X|\leq 15d_\TBR(S,S')-9$.  To establish that the bound is tight, it is \purple{sufficient to} choose $S$ and $S'$ such that they are subtree and chain reduced, and  $|X|=15d_{\TBR}(S,S')-9$. For  $k\geq 2$, set $S=S_k$ and $S'=S_k'$, where $S_k$ and $S_k'$ are the two trees shown in Figure~\ref{fig:TBR-SC-ex}. By construction, we have $|X|=15k-9$. Moreover, by Lemma~\ref{l:TBRsc=k}, we have $k=d_\TBR(S,S')$. The proof for when $S$ and $S'$ are subtree, chain, and cluster reduced can be established analogously by setting $S$ and $S'$ to be the two trees that are shown in Figure~\ref{fig:TBR-SCC-ex} for  $k\geq 4$ and considering Lemma~\ref{l:TBRscc=k} instead of Lemma~\ref{l:TBRsc=k}. \qed
\end{proof}

\blue{The next corollary is an immediate consequence of the last theorem.}

\begin{corollary}\label{c:TBR-kernel-tight}
The linear kernel as presented in Theorem~\ref{t:tbr-kernel} is tight.
\end{corollary}

\section{Tight kernels for computing the rooted variant of the minimum hybridization problem}

In this section, we turn to rooted phylogenetic trees and networks, and show that a previously published kernelization result that is concerned with the rooted \polish{analogue} of UHN is also tight. Before formally stating the problem, we introduce some new definitions some of which \polish{are} the rooted versions of their counterparts introduced in Section~\ref{sec:prelim}.

A {\it rooted binary phylogenetic network} $N$ on $X$ is a rooted acyclic digraph \blue{with no edges in parallel and satisfying} the following properties:
\begin{enumerate}[(i)]
\item the (unique) root $\rho$ has out-degree two,
\item the set $X$ \polish{labels}
the set of vertices of out-degree zero, each of which has in-degree one \polish{(i.e. the leaves)}, and
\item all other vertices either have in-degree one and out-degree two, or in-degree two and out-degree one.
\end{enumerate}
For two vertices $u$ and $v$ in $N$, we say that $u$ is a {\it parent} of $v$ and $v$ is a {\it child} of $u$ if $(u,v)$ is an edge in $N$. For a leaf $\ell$, we denote its unique parent by $p_\ell$. Furthermore, a vertex of in-degree two and out-degree one is called a {\it reticulation}, and we use \polish{$r(N)$} to denote the number of reticulations in $N$. Lastly, $N$ is called a {\it rooted binary phylogenetic tree} on $X$ if $\polish{r(N)}=0$. 

Let $T$ be a rooted binary phylogenetic tree on $X$, and let $Y\subset X$. A subtree of $T$ is {\it pendant} in $T$ if it can be obtained from $T$ by deleting a single edge. Now, for $n\geq 2$, let $C=(\ell_1,\ell_2,\ldots,\ell_n)$ be a sequence of distinct leaves in $X$. We call $C$ an \purple{$n$-{\it chain}} of $T$ if $p_{\ell_1}=p_{\ell_2}$ or $p_{\ell_1}$ is a child of $p_{\ell_2}$ and, for all $i\in\{2,3,\ldots,n-1\}$, we have \polish{that} $p_{\ell_i}$ is a child of $p_{\ell_{i+1}}$. Furthermore, we say that $Y\subset X$ is a {\it non-trivial cluster} of $T$ if $|Y|\geq 2$ and there exists a vertex in $T$ whose set of descendants is precisely $Y$.

Now, let $N$ be a rooted binary phylogenetic network on $X$, and let $T$ be a rooted binary phylogenetic tree on $X$. We say that $T$ is {\it displayed} by $N$ if $T$  can be obtained from a subtree of $N$ by suppressing vertices with in-degree one and out-degree one. Furthermore, for two rooted binary phylogenetic trees $T$ and $T'$ on $X$, we set
$$h(T, T') = \min_ N\{r(N)\},$$ where the minimum is taken over all rooted binary phylogenetic networks on $X$ that display $T$ and $T'$. 
Historically, this number is referred to as the \polish{{\it hybridization number}} for $T$ and $T'$ (see e.g. \cite{bordewich2007computing,vanIersel20161075,ierselLinz2013,whidden2013fixed} and references therein).\\

We now formally state the rooted version of UHN.\\

\noindent{\sc Rooted-Hybridization-Number (RHN)}\\
\noindent{\bf Input.} Two rooted binary phylogenetic trees $T$ and $T'$ on $X$.\\
\noindent{\bf Output.} A rooted binary phylogenetic network $N$ on $X$ that displays $T$ and $T'$ and such that $r(N)=h(T,T')$.\\

\begin{figure}[!ht]
\center
\includegraphics[width=11.3cm]{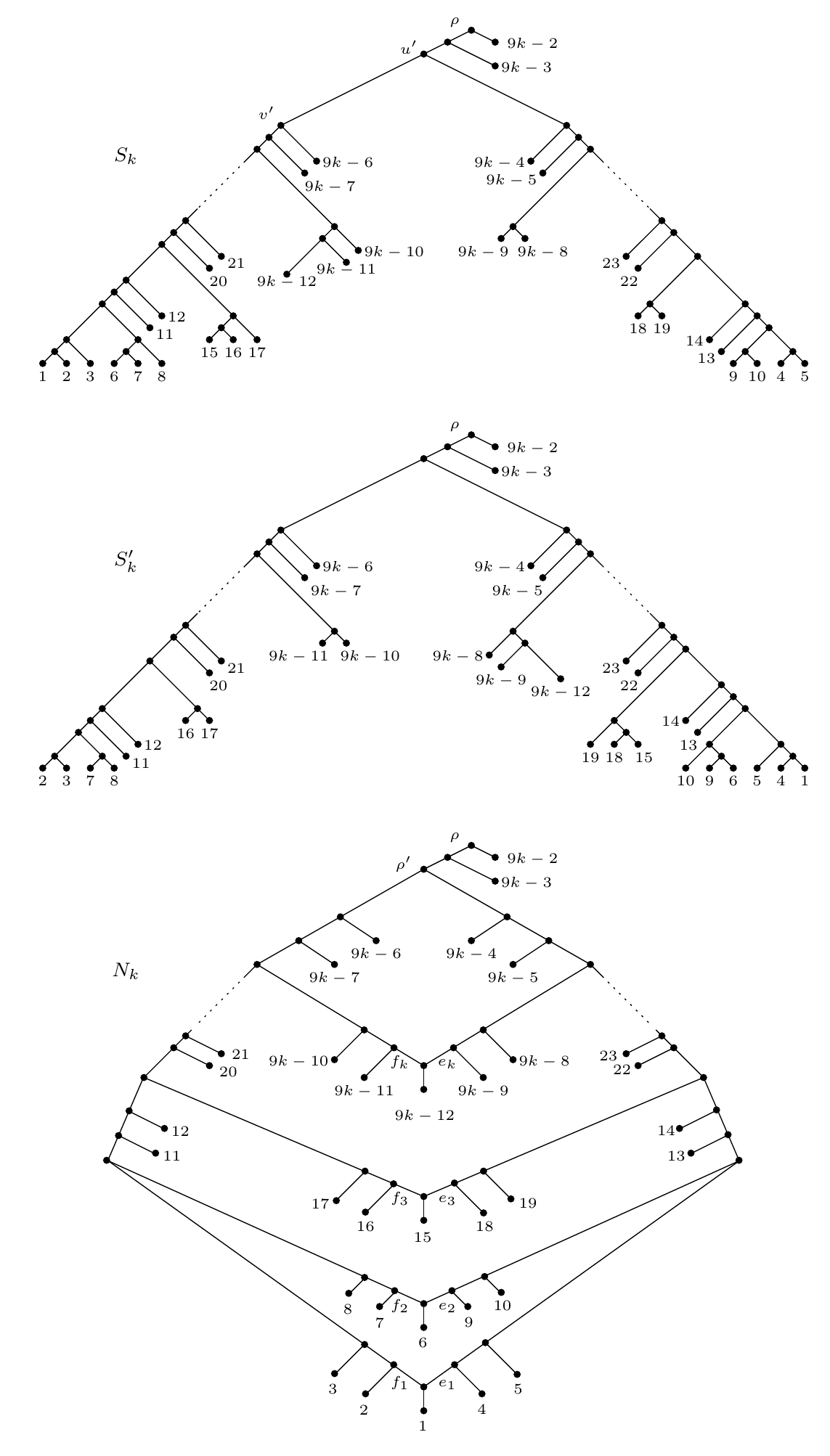}
\caption{The rooted binary phylogenetic trees $S_k$ and $S_k'$ as well as the rooted binary phylogenetic network that are used to show that the upper bound given in Lemma~\ref{l:r-hybrid} is tight for a pair of subtree and chain reduced trees. The edge \polish{$(u',v')$}
of $S_k$ is used to argue that $k= h(S_k,S_k')$ in the proof of Lemma~\ref{l:hybrid=k}.}
\label{fig:RHN-ex}
\end{figure}

Solving RHN for a pair of rooted binary phylogenetic trees is NP-hard~\cite{bordewich2007computing} and fixed-parameter tractable~\cite{sempbordfpt2007}, when parameterized by the \polish{hybridization number}. To establish a fixed-parameter tractability result, the authors of the latter paper used rooted variants of the subtree and chain reduction to kernelize the problem. Without detailing the 
reductions for two rooted binary phylogenetic trees, we next give  definitions of what it means for two rooted binary phylogenetic trees to be reduced under any of the two reductions. Specifically, for two rooted binary phylogenetic trees $S$ and $S'$ on $X$, we say that $S$ and $S'$ are {\it subtree reduced} if they do not have a common pendant subtree with at least two leaves and {\it chain reduced} if they do not have a common $n$-chain with $n\geq 3$. By building on the results of~\cite{sempbordfpt2007}, the authors of~\cite{approximationHN} improved the kernel size by using (rooted) generators. In the language of our paper, they established the following lemma.

\begin{lemma}\label{l:r-hybrid}\cite[Theorem 3.2]{approximationHN}
Let $S$ and $S'$ be two rooted binary phylogenetic trees on $X$ that are subtree and chain reduced
\polish{and $h(S, S') \geq 1$}. Then $|X|\leq 9h(S,S')-2$.
\end{lemma}

We next show, by \polish{describing a specific family of pairs of rooted binary phylogenetic trees}, that the bound given in Lemma~\ref{l:r-hybrid} is tight. For $k\geq 1$, consider the two rooted binary phylogenetic trees $S_k$ and $S_k'$ on $X_k$ that are shown in Figure~\ref{fig:RHN-ex}. By construction, we have $|X_k|=9k-2$; particularly $|X_1|=7$ and $|X_k|=7+9(k-1)$ for each $k\geq 2$. Moreover, it is easy to see that $S_k$ and $S_k'$ do not have any  common pendant subtree with at least two leaves or any common $n$-chain with $n\geq 3$. Thus, $S_k$ and $S_k'$ are subtree and chain reduced.

\begin{lemma}\label{l:hybrid=k}
For $k\geq 1$, let $S_k$ and $S_k'$, be the two rooted binary phylogenetic trees on $X_k$ that are shown in Figure~\ref{fig:RHN-ex}. Then $h(S_k,S_k')=k$.
\end{lemma}

\begin{proof}
Let $N_k$ be the rooted binary phylogenetic network that is shown in Figure~\ref{fig:RHN-ex}. Note that the leaf set of $N_k$ is $X_k$. Moreover, observe that $S_k$ can be obtained from $N_k$ by deleting the set $\{e_1,e_2,\ldots, e_k\}$ of edges and, subsequently, suppressing all resulting vertices with in-degree one and out-degree one. Similarly, $S_k'$ can be obtained from $N_k$ by deleting the set $\{f_1,f_2,\ldots, f_k\}$ of edges and, subsequently, suppressing all resulting vertices with in-degree one and out-degree one. Hence, $N_k$ displays $S_k$ and $S_k'$. Since $r(N_k)=k$, we have 
\begin{equation}\label{eq:three}
h(S_k,S_k')\leq r(N_k)=k.
\end{equation}
We complete the proof by showing that $k\leq h(S_k,S_k')$. Let $N_k^*$ be a rooted binary phylogenetic network on $X_k$ such that $r(N_k^*)=h(S_k,S_k')$. Furthermore, let $\bar{S}_k$ and $\bar{S}_k'$ be the two unrooted binary phylogenetic trees on $X_k$ that are  obtained from $S_k$ and $S_k'$, respectively, by suppressing the root $\rho$ and ignoring the directions on the edges. Then $\bar{S}_k$ and $\bar{S}_k'$ are displayed by the unrooted binary phylogenetic network $\bar{N}^*_k$ on $X_k$ that is obtained from $N_k$ by suppressing its root and ignoring the directions on the edges. As $r(N_k^*)=r(\bar{N}^*_k)$, it follows that $h^u(\bar{S}_k,\bar{S}_k')\leq h(S_k,S_k')$. To show that $k\leq h^u(\bar{S}_k,\bar{S}_k')$, we use the same approach as in the second half of the proof of Lemma~\ref{l:TBRsc=k}, where the edge $(u',v')$ as depicted in $S_k$ of Figure~\ref{fig:RHN-ex} plays the role of the edge $\{u,v\}$. Hence, we have $$k=|1-(k+1)|\leq d_\MP(\bar{S}_k,\bar{S}_k')\leq d_\TBR(\bar{S}_k,\bar{S}_k')= h^u(\bar{S}_k,\bar{S}_k')\leq h(S_k,S_k')$$ which, in combination with Equation~\ref{eq:three}, establishes the lemma. \qed
\end{proof}

The main result of this section is the following theorem whose proof can be established in the same way as the proof of Theorem~\ref{t:TBR-tight}.

\begin{theorem}\label{t:running-out-of-labels}
Let $S$ and $S'$ be two rooted binary phylogenetic trees on $X$ \polish{and $h(S, S') \geq 1$}. If $S$ and $S'$ are subtree and chain reduced, then $|X|\leq 9h(S,S')-2$ is a tight bound. 
\end{theorem}

\noindent Similar to UHN and Corollary~\ref{c:TBR-kernel-tight}, the last theorem implies that the upper bound on the size of a kernel for RHN, as established in~\cite[Theorem 3.2]{approximationHN}, is also tight.\\

We now turn to the rooted version of the cluster reduction~\cite{baroni2006hybrids}. Informally, this reduction breaks an instance, say $T$ and $T'$, of RHN into a number of smaller tree pairs such that the sum of the \polish{hybridization number} over all tree pairs equates to this number for $T$ and $T'$. In what follows, we say that two rooted binary phylogenetic trees are {\it cluster reduced} if they do not have any non-trivial cluster in common. Observe that the two phylogenetic trees $S_k$ and $S_k'$ as shown in Figure~\ref{fig:RHN-ex} have cluster $X_k-\{9k-3,9k-2\}$ in common. By deleting $9k-3$ and $9k-2$ from the trees and network of  Figure~\ref{fig:RHN-ex} and their respective parents, we obtain two rooted binary phylogenetic trees, say $R_k$ and $R_k'$, and a rooted binary phylogenetic network, say $M_k$.  Clearly, $R_k$ and $R_k'$  are subtree, chain, and cluster reduced. Furthermore, $M_k$ displays $R_k$ and $R_k'$ and $r(M_k)=r(N_k)=k$ for any $k\geq 1$. 
Reworking the proof of~\cite[Theorem 3.2]{approximationHN}, we note that the edge side incident with $\rho$ (this is a particular side of the generator used in that proof, corresponding to the path from $\rho$ to $\rho'$ as shown in Figure 3) cannot be decorated with any leaf. This is because any two distinct trees displayed by the network would then have a non-trivial common cluster. Hence, the counting argument from~\cite[Theorem 3.2]{approximationHN} goes through, with the exception that leaves $9k-3$ and $9k-2$ are no longer present in the network. This yields the following lemma: 

\begin{lemma}\label{l:r-hybrid-scc}
Let $S$ and $S'$ be two rooted binary phylogenetic trees on $X$ that are subtree, chain and cluster reduced \polish{and $h(S, S') \geq 1$}. Then $|X|\leq 9h(S,S')-4$.
\end{lemma}

Moreover, by repeating the argument to establish Theorem~\ref{t:running-out-of-labels} but using $R_k$, $R_k'$, and $M_k$ instead of $S_k$, $S_k'$, and $N_k$ as shown in Figure~\ref{fig:RHN-ex}, it  follows that the bound given in Lemma~\ref{l:r-hybrid-scc} is tight.

\section{Discussion}
Following the \blue{results} in this article, the algorithmic state of knowledge about computation of TBR distance can be summarized as follows. The problem is NP-hard, but permits a polynomial-time 3-approximation \cite{chen2015parameterized}, a branching FPT algorithm with running time $O( 3^k \cdot \text{poly}(n) )$ (where $k$ is the TBR distance and $n=|X|$) \cite{whidden2013fixed}, a kernel of size $15k-9$ and an exponential-time algorithm with running time  $O( 2.619^n \cdot \text{poly}(n))$ \cite{kelk2017note}.
An interesting consequence of the strengthened $15k-9$ bound is that results which indirectly use the size of the TBR kernel to compute bounds, automatically improve. One concrete example of this is the result in \cite{kelk2016reduction} which proves that after application of subtree and chain reduction rules, reduced instances of $d_\MP$ contain at most $\frac{4}{3} \cdot 28 \cdot d_\TBR$ taxa. The 28 now improves to 15. \blue{Our result implies a similar improvement on the kernel for computing the so-called subtree prune and regraft (SPR) distance~\cite{whidden2018calculating}.}

Returning to TBR distance, the natural question is whether through the introduction of new polynomial-time reduction rules (in addition to subtree, chain and cluster reduction rules) the size of the kernel can be further reduced, and if so, what the limit is of such an approach. This ties in with the FPT literature on (complexity-theoretic) lower bounds on kernel size (see e.g. \cite{bodlaender2014kernelization}), which have not yet been explored in the phylogenetics literature. Towards an easy lower bound we note that, if computation of TBR distance is APX-hard, then there exists a constant $c > 1 $ such that a polynomial-time $c$-approximation for TBR distance is not possible (assuming P $\neq$ NP). Such a result would exclude the existence of a kernel of size $c \cdot k$ for TBR distance (assuming P $\neq$ NP). This is because the TBR distance of two trees on $n$ taxa is at most $O(n)$ \cite{AllenSteel2001}; so simply returning a trivial solution for the kernelized instance would yield a $c$-approximation.\footnote{\blue{Note that the kernel for RHN is weighted~\cite{approximationHN} and, so, returning a trivial solution for a kernelized instance of RHN does not yield a $c$-approximation for this problem.}} However, although it is likely that computing the TBR distance is APX-hard, to the best of our knowledge the result has never been proven. This is an interesting hole to close in the literature.

Beyond TBR distance we can ask whether existing bounds on kernels for other phylogenetic distances and incongruency measures are tight and, if not, whether they can be improved. Many such kernelizations use subtree reductions and variations of chain reductions. The reductions typically have a common core but details differ from case to case depending on the specific combinatorial nature of the problem at hand: there are many subtle differences between TBR distance, SPR distance \cite{bonet2010complexity,bordewich2005computational,whidden2018calculating}, hybridization number \cite{sempbordfpt2007,van2018unrooted,ierselLinz2013} and agreement forests \cite{shi2018parameterized,whidden2013fixed}, for example. Other relevant factors include whether the input trees are unrooted or rooted; whether the input trees are binary or non-binary and the number of trees allowed in the input (see earlier references and \cite{vanIersel20161075}). In obtaining the tight $15k-9$ bound we were greatly helped by our ability to re-formulate the problem as a phylogenetic network construction problem, which in turn allowed us to make use of generators. Interestingly, the generators not only helped us improve the upper bound, they also gave strong hints concerning the topology of tight instances. Once discovered, we could use \blue{the maximum parsimony distance} to argue lower bounds on $d_\TBR$ distance. In how far do these three ingredients exist simultaneously for other problems and, where they do not exist, in which direction do we have to advance our knowledge to obtain tight bounds on kernel sizes? 


\vspace{2cm}
\noindent{\bf Acknowledgements.}
Simone Linz was supported by the New Zealand Marsden Fund.  Both authors would like to thank the Lorentz Center in the Netherlands for hosting the workshop \emph{Distinguishability in Genealogical Phylogenetic Networks}, where this work was initiated.

\bibliography{TBR-Kernel}{}
\bibliographystyle{plain}

\end{document}